\documentclass[a4paper,UKenglish,cleveref, autoref, thm-restate, final]{lipics-v2021}

\usepackage{mathtools}
\usepackage{latexsym}
\usepackage{amsmath}
\usepackage{amssymb}
\usepackage{amsthm}
\usepackage{epsfig}
\usepackage{color}
\usepackage{hyperref}
\usepackage{pstricks}
\usepackage{pst-node}
\usepackage{pst-tree}
\usepackage{pst-coil}
\usepackage{multirow}
\usepackage{graphics}
\usepackage{tikz}
\usetikzlibrary{intersections}
\usetikzlibrary{decorations.markings}
\usetikzlibrary{arrows, automata, positioning,calc}
\usetikzlibrary{decorations.pathmorphing}
\usetikzlibrary{decorations.pathreplacing}

\newcommand{\RP}{\mathbb{R}_{\ge 0}}
\newcommand{\R}{\mathbb{R}}
\newcommand{\N}{\mathbb{N}}

\newcommand{\locs}{L} % location set 
\newcommand{\loc}{\ell} % one location 
\newcommand{\clocks}{\mathcal{X}} % clock variables
\newcommand{\guard}{\varphi} 
\newcommand{\reset}{\lambda}
\newcommand{\val}{\nu} % clock valuation
\newcommand{\cval}{{\mathrm{val}}}

\newcommand{\true}{\texttt{true}} % TRUE

\newcommand{\tuple}[1]{\langle #1 \rangle}

\newcommand{\edges}{E}

\theoremstyle{definition}

\pdfoutput=1 

\title{Reachability for Multi-Priced Timed Automata with Positive and 
  Negative Rates}

\titlerunning{Reachability for Multi-Priced Timed Automata with Positive and 
  Negative Rates}

\author{Andrew Scoones}{Department of Computer Science, University of Oxford, United Kingdom}{andrew.scoones@cs.ox.ac.uk}{https://orcid.org/0000-0002-0610-7998}{Supported by UKRI Frontier Research Grant EP/X033813/1}

\author{Mahsa Shirmohammadi}{CNRS, IRIF, Universit\'{e} of Paris
  Cit\'{e},
  France}{mahsa@irif.fr}{https://orcid.org/0000-0002-7779-2339}{Supported
by VeSyAM (ANR-22- CE48-0005)}

\author{James Worrell}{Department of Computer Science, University of
  Oxford, United
  Kingdom}{jbw@cs.ox.ac.uk}{https://orcid.org/0000-0001-8151-2443}{Supported
  by UKRI Frontier Research Grant EP/X033813/1}

\authorrunning{A.\ Scoones, et al.} %

\Copyright{Andrew
  Scoones, Mahsa Shirmohammadi, James Worrell} %



\nolinenumbers

\begin{document}
\maketitle

\begin{abstract}
  Multi-priced timed automata (MPTA) are timed automata with observer
  variables whose derivatives can change from one location to another.
  Observers are write-only variables, that is, they do not affect the control
  flow of the automaton; thus MPTA lie between timed and hybrid
  automata in expressiveness.  Previous work considered observers with
  non-negative slope in every location.  In this paper we treat
  observers that have both positive and negative rates.  Our
  main result is an algorithm to decide a gap version of the
  reachability problem for this variant of MPTA.  We translate the
  gap reachability problem into a gap satisfiability problem for mixed
  integer-real systems of nonlinear constraints.  Our main technical
  contribution -- a result of independent interest -- is a procedure
  to solve such contraints via a combination of branch-and-bound
  and relaxation-and-rounding.
  \end{abstract}

  % ==========introduction============

  \section{Introduction}\label{sec:intro}
  Timed automata~\cite{AlurDill94} are a widely studied model of
  real-time systems that extend classical finite state-automata with
  real-valued variables, called \emph{clocks}, that evolve with
  derivative one and which can be queried and reset along transitions.
  \emph{Multi-Priced Timed Automata} (MPTA)~\cite{Bouyer08b,
    Brihayebruyereraskin06,SwaminathanF09,TollundJNTL24} further
  extend timed automata with write-only variables, called
  \emph{observers}, that have a non-negative slope that can change
  from one location to another.  Such variables can model the
  accumulation of costs or the use of resources along a computation,
  such as energy and memory consumption in embedded systems, or
  bandwidth in communication networks.  For this reason MPTA are
  widely used to model multi-objective real-time optimisation
  problems~\cite{BouyerFLM11}.

  While observers exhibit richer dynamics than clocks, they may not be
  queried while taking edges.  Thus MPTA lie between timed automata
  (for which reachability is decidable) and linear hybrid automata
  (for which reachability is undecidable~\cite{Henzinger98}).
  A natural class of verification problems for MPTA concerns
  reachability subject to constraints on the observers.  A simple
  variant is the \emph{Domination Problem}, which asks to reach a
  location subject to upper bounds on each observer.  Here one can
  think of the constraints as representing upper bounds on accumulated
  costs or resources.  The Domination Problem was shown decidable
  in~\cite{Larsen08} using well-quasi-orders and was later shown to be
  PSPACE-complete in~\cite[Theorem 4]{FranzleSSW22}.

  A more expressive version of the Domination Problem partitions the
  set of observers into \emph{cost variables} and \emph{reward
    variables} and asks to reach a location subject to upper bounds on
  costs and lower bounds on rewards.  This variant is, unfortunately,
  undecidable.  However it is shown in~\cite[Theorem 6]{FranzleSSW22}
  that a gap version of the problem---called the \emph{Gap Domination
    Problem}---is decidable.  In the Gap Domination Problem the input
  additionally contains a slack $\varepsilon>0$.  The objective is to
  distinguish the case that the constraints on the observers can be
  satisfied with slack $\varepsilon$ from the case in which they
  cannot be satisfied at all.  In general, gap problems are decision
  versions of approximation problems~\cite[Chapter 18.2]{AB06}.
  Decidability of the Gap Domination Problem implies that the Pareto
  curve of undominated reachable cost vectors can be computed to
  arbitrary precision (cf.~\cite{DiakonikolasY09}).

  The objective of this paper is to address a more expressive variant
  of MPTA than hitherto considered: namely those in which observers
  can have both positive and negative rates.  Alternatively, and
  equivalently, one can consider MPTA with nonnegative rates, but in
  which one allows reachability specifications to contain constraints
  on the \emph{difference} between two observers rather than just
  threshold constraints that compare observers to constants. Indeed,
  this extension is motivated by the desire to measure net resource
  use along computations.  In this more general setting, the
  Domination Problem, of course, remains undecidable; one moreover
  loses monotonicity properties on which previous positive
  decidability results rely, including the decision procedure for the
  Gap Domination Problem given in~\cite[Theorem 15]{FranzleSSW22}.
  The main result of this paper is to establish decidability (in
  nondeterministic exponential time) of the Gap Domination Problem in
  the presence of positive and negative rates via a new decision
  procedure.

We start by recalling a result of~\cite{FranzleSSW22} that
characterises the set of all reachable observer values for a given
MPTA via a system of mixed integer-real nonlinear constraints.  Our
main technical contribution, which is of independent
interest, shows how to solve a gap version of the satisfiability
problem for such systems of constraints.  Our method involves a 
combination of relaxation-and-rounding and branch-and-bound that
relies on Khinchine's Flatness Theorem from Diophantine
approximation.  We formulate a relaxation of the system of constraints
such that a solution to the relaxed version can be rounded to a
solution of the original problem, while unsolvability of the relaxed
version permits a branch-and-bound step that eliminates a variable
from the original system of constraints.

Systems of non-linear constraints over integer and real variables
appear in many different domains and are widely studied, although
typically not from the point of view of decidability since most
classes of problems with unbounded integer variables are
undecidable~\cite{hemmecke2010nonlinear}.  Other
than~\cite{FranzleSSW22}, we are not aware of previous work on the gap
problem considered here.  Kachiyan and
Porkolab~\cite{Khachiyan2000IntegerOO} showed that it is decidable
whether a convex semialgebraic set contains an integer point; however
we work with non-convex sets.

In this paper we consider MPTA with arbitrarily many observers.  There
is a significant literature and mature tool support concerning the
special case of MPTA with a single observer, which are variously
called Priced Timed Automata or Weighted Timed Automata.  In this
case, the optimal cost to reach a given location is
computable~\cite{Alur01PTA,Behrmann01,Larsen01}.  In the case of one
cost and one reward observer, one can also compute the optimal
reward-to-cost ratio in reaching a given location~\cite{Bouyer08b}.
The preceding results use the so-called \emph{corner-point
  abstraction}, which is insufficient for multi-objective model
checking.  Instead, the present paper implicitly relies on the
\emph{simplex-automaton abstraction}, introduced
in~\cite{FranzleSSW22}, which underlies the non-linear constraint
problems that are the subject of our main results.  All previously
mentioned works involve observers that evolve linearly with time.
Observer variables that vary non-linearly with time are considered
in~\cite{BhaveKT16}.  In the non-linear setting the optimal cost
reachability problem is undecidable in general.  Another variant, this
time towards greater simplicity, is to consider observers that are
only updated through discrete transitions~\cite{Zhang17}.

%Finally, the Domination Problem that we consider in this paper is
%decidable for at most three observers, by reduction to the problem of
%solving a single quadratic Diophantine equation~\cite[Theorem
%5]{FranzleSSW22}.

\section{Automata and Decision Problems}
\subsection{Multi-Priced Timed Automata}
\label{sec:MPTA}

Let $\RP$ denote the set of non-negative real numbers.  Given a
set~$\clocks=\{x_1,\ldots,x_n\}$ of \emph{clocks}, the set
$\Phi(\clocks)$ of \emph{clock constraints}\label{clock-constraint} is
generated by the grammar \[ \varphi ::= \true \mid x\leq k \,\mid\,
x\geq k \,\mid\, \varphi \wedge \varphi \, , \]where $k \in \N$ is a
natural number and $x\in \clocks$.  A \emph{clock valuation} is a
mapping~$\val: \clocks \to \RP$ that assigns to each clock a
non-negative real number.  We denote by $\boldsymbol 0$ the valuation such
that~$\boldsymbol 0(x)=0$ for all clocks $x\in \clocks$.  We write
$\val\models\varphi$ to denote that~$\val$ satisfies the
constraint~$\guard$.  Given $t\in\RP$, we let $\val+t$ be the clock
valuation such that~$(\val+t)(x)=\val(x)+t$ for all
clocks~$x\in\clocks$.  Given $\reset\subseteq\clocks$, let
$\val[\reset\leftarrow 0]$ be the clock valuation such that
$\val[\reset\leftarrow 0](x)=0$ if~$x\in\reset$, and
$\val[\reset\leftarrow 0](x)=\val(x)$ otherwise.

A \emph{multi-priced timed automaton} (MPTA)
$\mathcal{A}=\tuple{\locs,\ell_0,L_f,\clocks,\mathcal{Y},\edges,R}$
comprises a finite set~$\locs$ of \emph{locations}, an \emph{initial
  location} $\ell_0\in L$, a set $L_f\subseteq L$ of \emph{accepting
  locations}, a finite set $\clocks$ of \emph{clock variables}, a
finite set $\mathcal{Y}$ of \emph{observers}, a set
$\edges\subseteq \locs\times \Phi(\clocks)\times 2^\clocks\times
\locs$ of \emph{edges}, and a \emph{rate function}
$R : L \rightarrow \mathbb{Z}^{\mathcal{Y}}$.  Here
$R(\loc)(y)$ is the derivative of the observer $y \in \mathcal Y$ in
location $\loc$.  Denote by $\| \mathcal A\|$ the length of the
description of $\mathcal A$, where all integers are written in binary.

A \emph{state} of $\mathcal{A}$ is a triple $(\loc,\nu,t)$ where $\loc$
is a location, $\nu$ a clock valuation, and~$t\in\RP$ is a \emph{time
  stamp}.  A \emph{run} of $\mathcal{A}$ is an alternating sequence of
states and edges
\[ \rho = (\ell_0,\nu_0,t_0) \stackrel{e_1}{\longrightarrow}
  (\ell_1,\nu_1,t_1)\stackrel{e_2}{\longrightarrow} \ldots
  \stackrel{e_m}{\longrightarrow} (\ell_m,\nu_m,t_m) \, ,\] where $t_0=0$,
$\nu_0=\boldsymbol{0}$,~$t_{i-1} \le t_i$ for all
$i\in\{1,\ldots,m\}$, and
$e_i =\tuple{\ell_{i-1},\varphi,\lambda,\ell_i} \in E$ is such that
$\nu_{i-1} +(t_i-t_{i-1})\models \varphi$ and~$\nu_i = (\nu_{i-1}+(t_i-t_{i-1}))[\lambda \leftarrow 0]$
for~$i=1,\ldots,m$.  The run is \emph{accepting} if $\ell_m\in L_f$.
The \emph{value} of such a run is a vector $\cval(\rho) \in \mathbb{R}^{\mathcal{Y}}$, defined by
$\cval(\rho) = \sum_{j=0}^{m-1} (t_{i+1}-t_i) R(\ell_i) \, .$
We refer to Figure~\ref{fig:automaton} for an example of an MPTA and
its operational semantics.

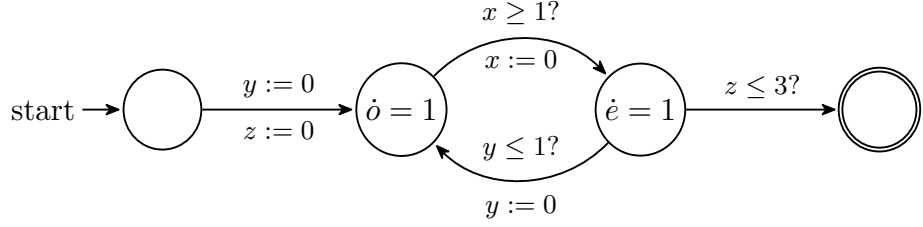
\begin{figure}
  \begin{center}
        \resizebox{4.8in}{1.2in}{%
\usetikzlibrary {arrows.meta,automata,positioning} 
\begin{tikzpicture}[->,>={Stealth[round]},shorten >=1pt, 
                    auto,node distance=2.7cm,on grid,semithick, 
                    inner sep=2pt,bend angle=45]
  \node[initial,state] (A)                    {}; 
  \node[state]         (B) [right=of A] {$\dot{o}=1$}; 
  \node[state]         (C) [right=of B] {$\dot{e}=1$}; 
  \node[accepting,state]         (D) [right=of C] {}; 

  \path [every node/.style={font=\footnotesize}]
  (A) edge node [above=0.5mm] {$y:=0$} node [below=0.5mm] {$z:=0$ }(B) 
  (B) edge [bend left] node [above=0.9mm] {$x\geq 1?$} node [below=0.5mm] {$x:=0$}   (C) 
 (C) edge [bend left] node [above=1.3mm] {$y\leq 1?$} node [below=0.9mm]
 {$y:=0$}  (B) 
  (C) edge node [above=0.5mm] {$z\leq 3?$}  (D); 
  
      \end{tikzpicture}}
    \end{center}
    \caption{ An MTPA with three clocks $x,y,z$ and two observer
      variables $o,e$, respectively standing for \emph{odd} and
      \emph{even}.  The observer variables have slope 0 unless
      otherwise indicated; thus $o$ aggregates the total dwell time in
      the \emph{odd} state and $e$ aggregates the total dwell time in
      the \emph{even} state.  An accepting run is
      completely determined by a sequence of nonnegative real numbers
      $d_0,\ldots,d_{2k}$, giving the respective delays between successive
      transitions.  Suppose we wish to reach the accepting state
      subject to the two objectives $e \geq 2$ and $o \geq 1$.  This
      is achieved, among others, by the run with sequence of time delays
      $\frac{2}{3}, \frac{1}{3},\frac{2}{3}, \frac{1}{3},\frac{2}{3},
      \frac{1}{3},\frac{2}{3}$ and the run with integer
      sequence of delays $1,0,1,0,1,1,0$ (and any convex combination
      of the two runs).  If the inequalities in the guards on $x$ and
      $y$ are replaced by equalities then the first run is the unique
      one realising the two given objectives.  In the case of
      so-called \emph{pure} reachability objectives, i.e., exclusively
      upper bound constraints or exclusively lower bound constraints
      on the observers,
      there is an explicit upper bound on the granularity of the
      delays in a run witnessing that the ojbective is realisable
      ($\frac{1}{3}$ in the present example)~\cite[Section
      6]{FranzleSSW22}.  This no longer holds in the case of
      reachability objectives that contain both upper and lower bounds
      on observers.}
\label{fig:automaton}
\end{figure}

\subsection{The Gap Domination Problem}

The \emph{Domination Problem} is as follows. Given an MPTA
$\mathcal{A}$ with set $\mathcal{Y}$ of observers and a
target~$\gamma\in \R^\mathcal{Y}$, decide whether there is an
accepting run $\rho$ of $\mathcal{A}$ such that
$\cval(\rho) \leq \gamma$ pointwise.

Our formulation of the Domination Problem involves a conjunction of
constraints of the form $y \leq c$, where $y\in\mathcal{Y}$ and
$c\in \mathbb{Q}$.  However such inequalities can encode
more general linear constraints of the
form~$a_1 y_1+\cdots+ a_ky_k \sim c$, where
$y_1,\ldots,y_k \in \mathcal{Y}$, $a_1,\ldots, a_k, c \in \mathbb{Z}$
and ${\sim} \in\{\leq,\geq,=\}$.  To do this one introduces a fresh
observer to denote each linear term~$a_1 y_1+\cdots+ a_ky_k$ (two
fresh observers are needed for an equality constraint).  For this
reduction it is crucial that we allow observers with negative rates.

The Domination Problem is PSPACE-complete for MPTA with positive rates
only~\cite[Theorem 11]{FranzleSSW22}, but is undecidable if negative rates are
allowed~\cite[Theorem 3]{FranzleSSW22}.  This motivates us to consider the
\emph{Gap Domination Problem}---a variant of the above problem in
which the input additionally includes a \emph{slack parameter}
$\varepsilon>0$.  If there is some run $\rho$ such that
$\cval(\rho)\leq \gamma-\varepsilon$ then the output should be
``dominated'' and if there is no run $\rho$ such
that~$\cval(\rho) \leq \gamma$ then the output should be ``not
dominated''.  In case neither of these alternatives hold (i.e.,
$\gamma$ is dominated but not with slack $\varepsilon$) then there is
no requirement on the output.  The Gap Domination Problem is the
decision version of the task of computing $\varepsilon$-approximate
Pareto curve in the sense of~\cite{DiakonikolasY09}.

The following proposition and (a generalisation
of~\cite[Propositions 6 and 7]{FranzleSSW22}), concerning the
structure of the set of reachable vectors of observer values, allows
us to reduce the Gap Domination Problem to a Diophantine problem.
Geometrically the proposition says that the set of reachable observer
vectors consists of a countable union of simplexes, where each simplex
is specified by its vertices -- a tuple of integer vectors -- and the
set of such tuples is semilinear.  The proposition is based on the
fact that if there are $d$ observers then any reachable observer valuation is a
convex combination of $d+1$ valuations that are respectively reached
along $d+1$ runs, all taking the same sequence of edges, in which all
transitions occur at integer time points (see~\cite{FranzleSSW22} for
details).

\begin{proposition}
  Let $\mathcal{A}$ be an MPTA with set of observers $\mathcal{Y}$
  having cardinality $d$.  Then there is a semilinear set
  $\mathcal{S}_{\mathcal{A}} \subseteq
  (\mathbb{Z}^{\mathcal{Y}})^{d+1}$ such that for every accepting run
  $\rho$ of $\mathcal{A}$ there exists
  $(\gamma_1,\ldots,\gamma_{d+1}) \in \mathcal{S}_{\mathcal A}$ for
  which $\mathrm{cost}(\rho)$ lies in the convex hull of
  $\{\gamma_1,\ldots,\gamma_{d+1}\}$.  Moreover
  $\mathcal S_{\mathcal A}$ can be written as a union of a collection
  of linear sets that can be computed in time exponential in
  $\|\mathcal A\|$ and each of which has a description length polynomial
  in $\| \mathcal A\|$.
\label{prop:convex-hull}
\end{proposition}
\begin{proof}
  The proposition was proved in~\cite{FranzleSSW22} under
  the assumption that observers have nonnegative slope.  The general
  case follows easily.  Indeed, given an arbitrary MPTA
  $\mathcal{A}=\tuple{\locs,\ell_0,L_f,\clocks,\mathcal{Y},\edges,R}$,
  we define a new MPTA $\mathcal{A}'$, differing from $\mathcal{A}$
  only in its set of observers and rate function, such that all
  observers in $\mathcal A'$ have non-negative rates.  The set of
  observers of $\mathcal{A}'$ is
  $\mathcal{Y}':=\{ y_{+},y_{-} : y \in \mathcal{Y}\}$ and the rate
  function $R'$ is given by
\[
  R'(y_+)(\ell) := \max(R(y)(\ell),0) \quad\text{and}\quad
  R'(y_{-})(\ell) := \max(-R(y)(\ell),0) 
\]
for all $y\in\mathcal{Y}$ and all $\ell \in \locs$.

Define
$\Phi : \mathbb{Z}^{\mathcal{Y'}} \rightarrow
\mathbb{Z}^{\mathcal{Y}}$ by
$\Phi(\gamma)(y) = \gamma(y_+)-\gamma(y_{-})$.  If a run $\rho$ of
$\mathcal A'$ has cost vector $\gamma$ then $\rho$ has cost vector
$\Phi(\gamma)$ considered as a run of $\mathcal A$.  Thus if we define
$\mathcal{S}_{\mathcal{A}}:=\Phi(\mathcal{S}_{\mathcal{A}'})$, where
$\Phi$ has been lifted pointwise to a linear map
$\Phi : (\mathbb{Z}^{\mathcal{Y'}})^{d+1} \rightarrow
(\mathbb{Z}^{\mathcal{Y}})^{d+1}$, then $\mathcal{S}_{\mathcal{A}}$
satisfies the requirements of the proposition.
  \end{proof}

  The following is immediate from Proposition~\ref{prop:convex-hull}.
  \begin{corollary}
    Given $\gamma \in \mathbb{R}^{\mathcal{Y}}$, there exists
    a run $\rho$ with $\cval(\rho)\leq \gamma$ if and only if the
    following mixed integer-real  system of non-linear inequalities has a solution.
\begin{equation}
\begin{array}{rclrcl}
  \lambda_1\gamma_1 + \cdots + \lambda_{d+1}\gamma_{d+1} &\leq & \gamma
  \qquad\qquad &
1 & = & \lambda_1 + \cdots + \lambda_{d+1} \\
\multicolumn{3}{l}{(\gamma_1,\ldots,\gamma_{d+1}) \in
  \mathcal{S}_{\mathcal{A}} }& 
                               0 & \leq & \lambda_1,\ldots,\lambda_{d+1} \\
  \multicolumn{3}{l}{\gamma_1,\ldots,\gamma_{d+1} \in
  \mathbb{Z}^{\mathcal{Y}}} &
\multicolumn{3}{l}{\lambda_1,\ldots,\lambda_{d+1} \in \mathbb{R}}                              
\end{array}
\label{eq:master-system}
\end{equation}
\label{cor:bilinear}
\end{corollary}  

In the following two sections we analyse systems of constraints of the
above form, obtaining a general result that allows us to solve the Gap
Domination Problem. 

  \section{Mixed Integer Bilinear Systems}
  \label{sec:mixed}
\subsection{The Satisfiability Problem}
  A \emph{mixed-integer bilinear (MIB) system} is a collection of
  constraints in integer variables $\boldsymbol{x}$ and real variables
  $\boldsymbol{y}$ of the form:
\begin{equation}
\begin{array}{ll}
& \boldsymbol{x}^\top A_i \boldsymbol{y} \leq b_i \qquad( i=1,\ldots,\ell)\\
& C\boldsymbol x \leq \boldsymbol  d\\
  & E\boldsymbol y \leq \boldsymbol f \\
& \boldsymbol x \in \mathbb{Z}^m, \boldsymbol y \in \mathbb{R}^n \,  .
\end{array}
\label{eq:bilinear}
\end{equation}
We assume that all constants in~\eqref{eq:bilinear} are integer; thus
if the system is satisfiable then there is a satisfying assignment in
which $\boldsymbol{y}$ is a rational vector.  We say that a satisfying
assignment has \emph{slack} $\varepsilon>0$ if 
$\boldsymbol{x}^\top A_i \boldsymbol {y} \leq b_i-\varepsilon$, for
$i=1,\ldots,\ell$.  Note that the slack requirement refers only to
the nonlinear constraints.  

We say that the system~\eqref{eq:bilinear} is \emph{bounded} if the
polyhedron
$\{ \boldsymbol y \in \mathbb{R}^n : E\boldsymbol y \leq \boldsymbol
f\}$ is bounded, i.e., is a polytope.  Crucially, the MIB systems
arising from multi-priced timed automata in
Corollary~\ref{cor:bilinear} are bounded.  Unfortunately, however, the
satisfiability problem for MIB systems is undecidable, even in the
bounded case.

\begin{proposition}
  The satisfiability problem for bounded mixed-integer bilinear
  systems is undecidable.
  \label{prop:undecidable}
\end{proposition}
\begin{proof}
  We reduce from the following version of Hilbert's 10th Problem
  (see~\cite[Proposition 1]{FranzleSSW22}):
  given a finite system $\mathcal S$ of equations in variables
  $x_1,\ldots,x_n$, with each equation either having the form 
  $x_i=x_j+x_k$ or $x_i=x_jx_k$, determine whether $\mathcal S$ has a
  solution in the set of strictly positive integers.
  
  The reduction involves transforming the system $\mathcal S$ into an
  equisatisfiable MIB system $\mathcal S'$ over a set of integer
  variables $x_0,\ldots,x_n \geq 0$ (i.e, the variables of
  $\mathcal S$ plus a new variable $x_0$) and real variables
  $y_1,\ldots,y_n \geq 0$.  The construction is such that every
  solution of $\mathcal S$ extends to a solution of $\mathcal S'$ and,
  conversely, every solution of $\mathcal S'$ restricts to a solution
  of $\mathcal S$.
  
  The system $\mathcal S'$ includes equations $x_0=1$ and $x_iy_i=1$
  for $i=1,\ldots,n$.  The linear equations $x_i=x_j+x_k$ from
  $\mathcal S$ are carried over to $\mathcal S'$ and, for each
  equation $x_i=x_jx_k$ in $\mathcal S$, we include an equation
  $(x_j+x_k)y_i = x_0(y_j+y_k)$ in $\mathcal S$.  The latter is
  equivalent to $\frac{x_j+x_k}{x_i} = \frac{1}{x_j}+\frac{1}{x_k}$ in
  the presence of the equations $x_iy_i=x_jy_j=x_ky_k=1$ and $x_0=1$,
  which in turn is clearly equivalent to $x_i=x_jx_k$.  By adding
constraints $0\leq y_i\leq 1$ for $i=1,\ldots,n$ we furthermore make
$\mathcal S'$ bounded without affecting the integrity of the
reduction.  
\end{proof}

\subsection{The Gap Satisfiability Problem}
In light of Proposition~\ref{prop:undecidable}, we introduce the 
following gap version of the satisfiability problem for 
MIB systems.  In this variant we seek a procedure 
that inputs $\varepsilon>0$ and a MIB system $\mathcal S$
in the form~\eqref{eq:bilinear} and returns either ``UNSAT'' or 
``SAT'' subject to the following requirements: 
\begin{enumerate}
\item If $\mathcal S$ has a satisfying assignment with slack 
  $\varepsilon$ then the output must be ``SAT''. 
\item If $\mathcal S$ is not satisfiable then the output must be 
  ''UNSAT''. 
\end{enumerate}  
Note that we place no restriction on the output in the case that $\mathcal S$ is 
satisfiable but with no satisfying assignment having slack 
$\varepsilon$. 

In Section~\ref{sec:dec-bounded} we will show that the Gap Satisfiability
Problem is decidable for bounded MIB systems.  The following
proposition shows the necessity of the boundedness hypothesis.

\begin{proposition}
The Gap Satisfiability Problem is undecidable for (unbounded) MIB systems. 
\label{thm:gap-undec}
\end{proposition}
\begin{proof}
  The proof is by reduction from the same variant of Hilbert's Tenth 
  Problem as in the proof of Proposition~\ref{prop:undecidable}.
  Recall that an 
  instance of this problem comprises a system $\mathcal S$ of equations in 
  positive-integer variables $x_1,\ldots,x_n$, with each equation 
  having the form either $x_i=x_j+x_k$ or $x_i=x_jx_k$, where 
  $i,j,k \in \{1,\ldots,n\}$.  Given such a system, we construct  an 
  MIB system $\mathcal S'$ over integer variables $x_0,\ldots,x_{n+1}$
  and real variables $y_0,\ldots,y_{n+1}$ such that every satisfying 
  assignment of $\mathcal S$ extends to a satisfying  assignment of 
  $\mathcal S'$ with slack 
  $\frac{1}{2}$ and every satisfying assignment of $\mathcal S'$
  restricts to a satisfying assignment of~$\mathcal S$. 

  We include the equations $x_0=1$ and $y_0=1$ in $\mathcal S'$. 
  Each linear equation $x_i=x_j+x_k$ in $\mathcal S$ is carried over 
  to $\mathcal S'$.  For each equation $x_i=x_jx_k$ in $\mathcal S$ we 
  include the inequality $|x_iy_0-x_jy_k| \leq \frac{1}{2}$ in 
  $\mathcal S'$.  We then add the following collection of constraints 
  to $\mathcal S'$ for all $i\in\{1,\ldots,n+1\}$ 
  that intuitively force $x_i$ and $y_i$ to 
  be very close together:
  \begin{enumerate}
 \item $|x_iy_0-x_0y_i| \leq 1$;
  \item $|x_{n+1}y_i - x_iy_{n+1}| \leq 1$;
  \item $x_{n+1}y_0 \geq 4(x_0+x_i)(y_0+y_i)+1$.
  \end{enumerate}

  A satisfying valuation of $\mathcal S$ can be extended to a
  valuation that satisfies $\mathcal S'$ with slack $\frac{1}{2}$ by
  setting $x_0:=1$,
  $x_{n+1}:=4\max_{i \in \{1,\ldots,n\} } (1+x_i)^2+1$, and $y_i:=x_i$
  for $i=0,\ldots,{n+1}$.

  Conversely, we claim that every satisfying valuation of $\mathcal
  S'$ (with no assumption on the slack)
  restricts to a satisfying valuation of $\mathcal S$.  Indeed, by 
  Item 2, above, for all $k \in \{1,\ldots,n\}$ we have 
  \[  |x_{n+1}(x_k-y_k) - x_k(x_{n+1}-y_{n+1}) | =
      |x_{n+1}y_k - x_ky_{n+1}| 
    \stackrel{(2)}{\leq} 1  .\]
By Items 1 and 3, this entails that for all $j \in \{1,\ldots,n\}$, 
\[ |x_k-y_k| \leq \frac{x_k|x_{n+1}-y_{n+1}|+1}{x_{n+1}}\stackrel{(1)}{\leq}
  \frac{x_k+1}{x_{n+1}} 
\stackrel{(3)}{\leq} \frac{1}{4(y_j+1)} \stackrel{(1)}{\leq} 
  \frac{1}{4x_j} \]
and hence $|x_jx_k-x_jy_k|\leq \frac{1}{4}$. 
Combined with $|x_i-x_jy_k|\leq \frac{1}{2}$ we conclude that 
$|x_i-x_jx_k|\leq \frac{3}{4}$ and hence $x_i=x_jx_k$. 
\end{proof}

It is shown in~\cite[Theorem 6]{FranzleSSW22} how to solve the Gap 
Satisfiabililty Problem for a subclass of MIB systems, which we here
call \emph{positive}.  A positive MIB system has the form
\[ \begin{array}{ll}
& \boldsymbol{x}^\top A_i \boldsymbol{y} \leq b_i \qquad(
                    i=1,\ldots,\ell_1)\\
&   \boldsymbol{x}^\top B_i \boldsymbol{y} \geq c_i \qquad(
                    i=1,\ldots,\ell_2)\\
  & D\boldsymbol y \leq \boldsymbol e \\
& \boldsymbol x, \boldsymbol y \geq \boldsymbol  0\\
  & \boldsymbol x \in \mathbb{Z}^m, \boldsymbol y \in \mathbb{R}^n \,  .
   \end{array} \]
 with all coefficients being non-negative rational.  This variant can
 be solved by a naive relaxing and rounding procedure, which does not
 require the boundedness assumption.  However, while
 suffiicent to handle MPTA with non-negative rates, this variant
 appears insufficient for the case of MPTA with both positive and
 negative rates.

\section{Decidability in  the Bounded Case}
\label{sec:dec-bounded}
  \subsection{Preliminaries}
  The following proposition on semilinear sets of 
  integers~\cite[Corollary 1]{pottier91} will be used on several 
  occasions below: 
  \begin{proposition}
    Consider a set 
    $S:=\{ \boldsymbol x \in \mathbb Z^m : A\boldsymbol x \leq 
    \boldsymbol b\}$, where the entries of $A$ and $\boldsymbol b$ are 
    integers of absolute value at most $H$ and the affine hull of $S$ has dimension $d$.  Then there exists 
    a finite set $B\subseteq \mathbb{Z}^m$ and a matrix 
    $P \in \mathbb Z^{m\times d}$ such that 
\[ S=L(B,P) := \{ \boldsymbol w + P\boldsymbol z : \boldsymbol w \in 
  B,\, \boldsymbol z\in \mathbb Z^d,\, \boldsymbol z\geq \boldsymbol 0\} \]
and the entries of $P$ and $\boldsymbol w$ have absolute value at 
most $(2+(m+1)H)^m$. 
\label{prop:semilinear}
\end{proposition}

We will also need the following result~\cite[Corollary 
3.1]{SchaeferS17} on semialgebraic sets of real numbers.  We assume 
that polynomials are written as lists of monomials with all integers, 
including exponents, written in binary. 

\begin{proposition}
Let $\{ f_i\}_{i \in I}$ be a family of polynomials in $n$ variables 
whose representation has total bit length at most $L$. 
Then the set $S:=\{ \boldsymbol x \in \mathbb{R}^n : \bigwedge_{i\in 
  I}  f_i \sim_i 0\}$, 
where ${\sim_i} \in \{<,=\}$, is either empty or contains a point of 
distance at most $2^{L^{8n}}$ to the origin. 
\label{prop:semi-alg}
\end{proposition}

For further analysis it will useful to transform the MIB problem to a
\emph{standard form}, shown in~\eqref{eq:standard} below.  In standard
form the only linear constraints on the integer variables are that
they be nonnegative.  Correspondingly we enrich the nonlinear
constraints, allowing them to contain an extra linear term in
$\boldsymbol y$.
\begin{equation}
\begin{array}{ll}
& \boldsymbol{x}^\top A_i \boldsymbol{y} +\boldsymbol{b}_i^\top 
                    \boldsymbol{y} \leq c_i  \quad(i=1,\ldots,\ell)\\
& D\boldsymbol{y}\leq \boldsymbol{e}\\
  & \boldsymbol{x} \geq \boldsymbol 0\\
  & \boldsymbol x \in \mathbb{Z}^m, \boldsymbol y \in \mathbb{R}^n \, . 
\end{array}
\label{eq:standard}
\end{equation}

The transformation of~\eqref{eq:bilinear} to standard form is based on
writing
$S:=\{ \boldsymbol x \in \mathbb{Z}^m : C\boldsymbol x \leq
\boldsymbol d\}$ as a semi-linear set $L(B,P)$, following
Proposition~\ref{prop:semilinear}, where $B \subseteq \mathbb{Z}^m$
and $P \in \mathbb{Z}^{m\times d}$ with $d$ the dimension of the
affine hull of $S$.  For each vector $\boldsymbol w \in B$ we can
apply the change of variables
$\boldsymbol{x}=P\boldsymbol{z}+\boldsymbol{w}$ to~\eqref{eq:bilinear}
to obtain a problem in standard form:  Thus we obtain a finite
collection of problems in standard form, whose solutions are in
one-one correspondence with the solutions of the original system~\eqref{eq:bilinear}.

\subsection{Relaxation and Rounding}
In this section we introduced a relaxed version of a bounded MIB
system, in which all variables range over the reals.  The relaxation
is such that a satisfying assignment to the relaxed problem can be
rounded to an integer solution of the original system, while
unsatisfiability of the relaxed version permits a branch-and-bound
step which leads to an equisatisfiable finite collection of MIB
instances in one fewer integer variable.

% In this section we present a recursive procedure to solve the Gap
% Satisfiability Problem on bounded MIB programs that works by
% relaxation and rounding.  The procedure attempts to answer ``SAT'' or
% ``UNSAT'' based on the satisfiability of a relaxed variant of the
% problem, in which all variables range over the reals.  The relaxation
% is such that a satisfying assignment to the relaxed problem can be
% rounded to an integer solution of the original system.  If the relaxed
% version is unsatisfiable then the algorithm branches, leading to a
% finite number of instances of a MIB problem in one fewer integer
% variable, on which the algorithm can be recursively applied.

The rounding is based on an application of the Flatness Theorem in
Diophantine approximation -- Theorem \ref{thm:flatness}, below.  To
state this result we first recall some standard terminology related to this.
Let $K\subseteq \mathbb{R}^n$ be a convex set and let
$\boldsymbol{u}\in\mathbb{Z}^n$.  Define the\emph{ width of $K$ with
  respect to $\boldsymbol{u}$} to be
\[ \mathrm{width}_{\boldsymbol{u}}(K) :=
\sup\{\boldsymbol{u}^\top(\boldsymbol{x}-\boldsymbol{y}) : 
\boldsymbol{x},\boldsymbol{y}\in K\} \, .\]
The \emph{lattice width} of $K$ is the minimum width in all directions:
\[ \mathrm{width}(K):=\min \{ \mathrm{width}_{\boldsymbol{u}}(K) : 
  \boldsymbol{u}\in\mathbb{Z}^n\setminus \{\boldsymbol{0}\} \} \, . \]

\begin{theorem}[Flatness Theorem]
  There exists a constant $\omega(n)$, depending only on $n$, such 
  every convex polyhedron $K\subseteq \mathbb{R}^n$ with 
  $\mathrm{width}(K) > \omega(n)$ contains an integer point. 
\label{thm:flatness}
\end{theorem}
The constant $\omega(n)$ in Theorem~\ref{thm:flatness} is called the
\emph{flatness constant}.  The best-known upper bound on
$\omega(n) =O(n^{3/2})$~\cite{banaszczyk}, although a linear upper
bound was conjectured in~\cite{kannan1988covering}.

We embark on a brief digression about the definability of the
lattice width for classes of polyhedral sets.
\begin{proposition}
  There is a quantifier-free formula in the theory of real closed
  fields, whose free variables respectively represent a matrix
  $A \in \mathbb R^{n\times m}$, vector
  $\boldsymbol b \in \mathbb R^n$, and scalar $c>0$, that expresses the
  property $\mathrm{width}_{\boldsymbol u}(P) \geq c$ where
  $P:=\{ \boldsymbol x \in \mathbb R^m : A\boldsymbol x \geq
  \boldsymbol b, \boldsymbol x \geq \boldsymbol 0\}$.
\end{proposition}
\begin{proof}
A necessary condition that $\mathrm{width}_{\boldsymbol u}(P) \geq c$
is that $P$ be non-empty and hence, since it lies in the positive
orthont, contain a vertex.  Now each vertex of $P$, being the
intersection of $n $ linearly independent bounding hyperplanes, has
the form $B^{-1} \boldsymbol b'$, where $B$ is a non-singular
$n\times n$ sub-matrix of $\begin{pmatrix}A \\ I_n \end{pmatrix}$,
where $I_n$ denotes the identity matrix of dimension $n$, and
$\boldsymbol b'$ is a
corresponding sub-vector of $\begin{pmatrix} \boldsymbol b\\
  \boldsymbol 0 \end{pmatrix}$.  Hence the vertices of $P$ are
definable by quantifier-free formulas.

Assume that $P$ contains a vertex.  Then
$\mathrm{width}_{\boldsymbol u}(P)$ is infinite if and only if either
$\boldsymbol u$ or $-\boldsymbol u$ lie in the recession cone of $P$,
for which a sufficient and necessary condition is that
$A\boldsymbol u \geq \boldsymbol 0$ or
$A\boldsymbol u \leq \boldsymbol 0$.
If $\mathrm{width}_{\boldsymbol u}(P)$ is finite then there exist two
vertices 
$\boldsymbol x_0,\boldsymbol x_1$ of $P$ such that 
$\mathrm{width}_{\boldsymbol u}(P) = \boldsymbol u^\top (\boldsymbol 
x_0 - \boldsymbol x_1)$.  
The proposition follows by combining the above observations.
\end{proof}

We now commence the detailed description of the relaxation
construction.  The input is a bounded MIB program $\mathcal S$ in
standard form~\eqref{eq:standard} and a slack $\varepsilon>0$.  Assume
that $\mathcal S$ has at least one non-linear constraint.  We start
with the observation that for a given $\boldsymbol y \in \mathbb R^n$
the system~\eqref{eq:standard} admits a solution
$\boldsymbol x \in \mathbb Z^m$ if and only if the polyhedral set
\begin{gather} P(\boldsymbol{y}):=\{ \boldsymbol{x} \in \mathbb{R}^m : 
   \boldsymbol x\geq \boldsymbol 0,\, \boldsymbol{x}^\top A_i \boldsymbol{y}+\boldsymbol b_i^\top 
   \boldsymbol y
   \leq c_i, \; i=1,\ldots,\ell \} \, , 
\label{eq:defP}
\end{gather}
contains an integer point.

Let $H$ be an upper bound of the absolute value of the integer
constants in the system~\eqref{eq:standard}.  
Since  $\mathcal S$ is bounded, by~\cite[Lemma 3.1.25]{grotschel} 
the set $\{ \boldsymbol{y}\in \mathbb{R}^n : D\boldsymbol{y}\leq 
\boldsymbol{e}\}$ is contained in the ball
of radius $\kappa_1:=m^{1/2} H^{(m^2+m)}$ centred at the origin.

For a matrix $A$, let $\| A\|$ denote the spectral norm.
Recall that  if $A$ has entries of absolute value at most $H$ and has
$m$ columns then $\|A\|\leq \sqrt{m}H$.  Now write 
     \begin{gather}
\delta:=\min(\delta_0,1),\quad\text{where }
       \delta_0 := \min \left \{ \frac{\varepsilon}{\|A_i\| \kappa_1} : 
    i=1,\ldots, \ell \right\}  \geq \frac{\varepsilon}{m^{1/2} H \kappa_1}
\label{def:delta}
\end{gather}
  and define 
$U:=\{ \boldsymbol{u}\in\mathbb{Z}^m \setminus \{\boldsymbol{0}\} : 2\delta \|\boldsymbol{u}\| <
\omega(m) \}$, where $\omega(m)$ is as in Theorem~\ref{thm:flatness}.
Write $U=\{\boldsymbol{u}_1,\ldots,\boldsymbol{u}_s\}$ and
consider the following \emph{relaxed system} $\mathcal S'$ of linear
and bilinear constraints in exclusively real variables (where the
notation $P(\boldsymbol y)$ is as in~\eqref{eq:defP}):
\begin{equation}
  \begin{array}{ll}
& \boldsymbol x^\top A_i  \boldsymbol {y} + \boldsymbol{b}_i^\top
                      \boldsymbol y \leq c_i- \varepsilon
                      \qquad( i=1,\ldots,\ell)\\[2pt]
    &   \mathrm{width}_{\boldsymbol u_j}(P(\boldsymbol y)) 
                                                         \geq \omega(m)    \qquad (j=1,\ldots,s)\\[2pt]
    & D\boldsymbol{y}\leq \boldsymbol{e},\;
    \boldsymbol x\geq \boldsymbol 1 \\[2pt]
    & \boldsymbol{x} \in   \mathbb{R}^m,\,  \boldsymbol y \in \mathbb{R}^n 
\end{array}
\label{eq:relaxed2}
\end{equation}

\begin{proposition}
  If the relaxed system $\mathcal S'$ is satisfiable, then so is
  the original system $\mathcal S$.
\label{prop:relaxed}
\end{proposition}
\begin{proof}
Let $\boldsymbol{x}^*,\boldsymbol{y}^*$
  be a solution of the system $\mathcal S'$, as  shown
  in~\eqref{eq:relaxed2}.  Consider the set $P(\boldsymbol y^*)$ as
  defined in~\eqref{eq:defP}.
  By construction we have
\begin{gather}
 \min_{\boldsymbol{u}\in U} \mathrm{width}_{\boldsymbol{u}}(
 P(\boldsymbol{y}^*) ) \geq \omega(m) \,  . 
\label{eq:width-small}
\end{gather}
But from the fact $\boldsymbol{x}^*$ satisfies each constraint
$\boldsymbol x^\top  A_i \boldsymbol{y}^* +\boldsymbol b_i^\top
\boldsymbol y^* \leq c_i$ with slack
$\varepsilon$ and that  $\boldsymbol x^* \geq \boldsymbol 1$, we see that the ball
$B_\delta(\boldsymbol{x}^*)$ is contained in $P(\boldsymbol{y}^*)$,
for $\delta$ as defined in~\eqref{def:delta}.  It follows that
\begin{eqnarray*}
  \mathrm{width}_{\boldsymbol{u}}(P(\boldsymbol{y}^*)) & \geq &
  2\delta\|u\| \\
  &\geq & \omega(m)
\end{eqnarray*}
for all $\boldsymbol{u} \not\in U$.  Together
with~\eqref{eq:width-small}, we have that
$\mathrm{width}(P(\boldsymbol{y}^*)) \geq \omega(m)$ and hence, by
Theorem~\ref{thm:flatness}, $P(\boldsymbol{y}^*)$ contains an integer
point.  This entails that the original system $\mathcal S$ is
satisfiable.
\end{proof}

\begin{proposition}
  If the relaxed system $\mathcal S'$ has no solution then every
  solution $\boldsymbol x^* \in \mathbb Z^m$ of the original system
  $\mathcal S$ that has slack $\varepsilon$ either has some component
  equal to zero or satisfies
  $|\boldsymbol{u}^\top \boldsymbol{x}^*| \leq \kappa_2$ for some
  $\boldsymbol u\in U$, where $\kappa_2$ is an explicit constant
  depending only on $\mathcal S$ and~$\varepsilon$.
\label{prop:split}
\end{proposition}
  \begin{proof}
    Assume that $\mathcal S'$ has no solution.  
    Let $\boldsymbol{x}^* \in \mathbb{Z}^m$ and
    $\boldsymbol{y}^* \in \mathbb{R}^n$ be a solution of $\mathcal S$
    with slack $\varepsilon$.  If some component of $\boldsymbol{x}^*$
    is zero then we are done, so we may suppose that
    $\boldsymbol x^* \geq \boldsymbol 1$.  By assumption,
    $\boldsymbol{x}^*,\boldsymbol{y}^*$ is not a solution of
    $\mathcal S'$ and so it must hold that
  \begin{gather}
 \min_{\boldsymbol{u}\in U} \mathrm{width}_{\boldsymbol{u}}(
 P(\boldsymbol{y}^*) ) < \omega(m) \,  , 
\label{eq:narrow}
\end{gather}
where  $P(\boldsymbol{y}^*)$ is as defined in~\eqref{eq:defP}.

Let $\boldsymbol{u} \in U$
be the vector achieving the minimum on the left-hand side of~\eqref{eq:narrow}.
We will exhibit an upper bound on $|\boldsymbol{u}^\top 
\boldsymbol{x}^*|$ that does not depend on~$\boldsymbol{y}^*$.  

Assume first that $P(\boldsymbol{y}^*)$ contains the origin.  Then by~\eqref{eq:narrow},
\[ |\boldsymbol{u}^\top \boldsymbol{x}^\ast| = |\boldsymbol u^\top
  (\boldsymbol x^\ast - \boldsymbol 0)|   \leq \omega(m) \, .\]

Assume now that $P(\boldsymbol{y}^*)$ does not contain the origin.  Let
$L$ be the line segment connecting the origin to $\boldsymbol{x}^*$,
and denote by $\boldsymbol{x}$ the point at which $L$ intersects the
boundary of $P(\boldsymbol{y}^*)$.  Then we have $\boldsymbol x^* -
\boldsymbol x = \lambda \boldsymbol x$ for some $\lambda > 0$.
Moreover, since $\boldsymbol x$ lies on the boundary of $P(\boldsymbol{y}^*)$
there exists $i_0 \in \{1,\ldots,\ell\}$ such that 
\begin{gather} \boldsymbol{x}^\top A_{i_0} \boldsymbol{y}^* + \boldsymbol
  b_{i_0}^\top \boldsymbol y^*=c_{i_0} \, ,
\label{eq:tight}
\end{gather}i.e., one of
inequalities that define $P(\boldsymbol{y}^*)$ is tight at
$\boldsymbol x$.  But since
$\boldsymbol x^*, \boldsymbol y^*$ satisfies $\mathcal S$ with slack
$\varepsilon$, we also have that
$(\boldsymbol{x}^*)^\top A_{i_0} \boldsymbol{y}^* +\boldsymbol
b_{i_0}^\top \boldsymbol y^* \leq c_{i_0} -\varepsilon$.  Subtracting
Equation~\eqref{eq:tight} from the previous inequality gives
\begin{eqnarray*}
-\varepsilon & \geq & (\boldsymbol x^*-\boldsymbol x)^\top A_{i_0}
                      \boldsymbol y^*\\
             &=& \lambda (\boldsymbol x^\top A_{i_0}\boldsymbol y^*)\\
  &=& \lambda (c_{i_0} - \boldsymbol b_{i_0}^\top \boldsymbol y^*) \, .
  \end{eqnarray*}
  Since $\varepsilon,\lambda>0$ this entails that
  $c_{i_0} - \boldsymbol b_{i_0}^\top \boldsymbol y^*<0$ and hence
  \begin{eqnarray}
\lambda^{-1} & \leq & \varepsilon^{-1} |c_{i_0} - \boldsymbol
                      b_{i_0}^\top \boldsymbol y^*| \notag \\
    & \leq &  \varepsilon^{-1}\left( |c_{i_0}|+\|\boldsymbol b_{i_0}\|
             \kappa_1\right)
             \label{eq:bound}
                               \end{eqnarray}
We deduce that
\begin{eqnarray*}
|\boldsymbol{u}^\top \boldsymbol{x}^*| 
 & \leq  &|\boldsymbol{u}^\top (\boldsymbol{x}^*-\boldsymbol{x}) |
    +|\boldsymbol{u}^\top \boldsymbol{x}| \\
  & =  & |\boldsymbol{u}^\top (\boldsymbol{x}^*-\boldsymbol{x})
         | \,  (1+\lambda^{-1})\\
  & \leq & \omega(m)
  \left(1+ \varepsilon^{-1}(|c_{i_0}|+\|\boldsymbol b_{i_0}
           \|\kappa_1) \right) \quad
           \text{ by~\eqref{eq:narrow} and~\eqref{eq:bound}.} 
  \end{eqnarray*}
  Thus, defining
  \begin{gather} \kappa_2:= \omega(m) \,  (1+
    H\varepsilon^{-1}(1+m^{1/2}\kappa_1) ) \,   ,
    \label{def:kappa2}
  \end{gather}
we have
  $|\boldsymbol{u}^\top \boldsymbol{x}^*| \leq \kappa_2$.

  In summary, we have that $|\boldsymbol{u}^\top \boldsymbol{x}^* |
  \leq \kappa_2$
  for every integer point $\boldsymbol x^*$ of $P(\boldsymbol{y}^*)$,
  as required in the proposition.
\end{proof}

\subsection{Decision Procedure}
In this section we describe a decision procedure for the Gap
Satisfiability Problem for bounded MIB systems.  This is a recursive
procedure based on the relaxation construction in the preceding
section.  We first present a conceptually simple version of the
procedure, with no complexity bound, and then give a more detailed
treatment from which bounds can be extracted.

\begin{theorem}
The Gap Satisfiability Problem is decidable for bounded MIB systems.
\label{thm:main-one}
\end{theorem}
\begin{proof}
  The procedure to solve the Gap Satisfiability Problem is as follows.
  Consider an instance of the problem, consisting of an MIB system in
  the form~\eqref{eq:standard} and slack $\varepsilon>0$.  If there
  are no non-linear constraints then the problem instance is just a
  system of linear inequalities in real and integer variables, whose
  satisfiability is straightforward to discern.  Thus we may assume
  that there is at least one non-linear constraint.  We construct the
  associated relaxed system $\mathcal S'$, which has the
  form~\eqref{eq:relaxed2}.  Using a decision procedure for the
  existential theory of real-closed fields we determine whether the
  system $\mathcal S'$ is satisfiable.

  If $\mathcal S'$ is satisfiable then
  Proposition~\ref{prop:relaxed} guarantees that the original MIB
  system $\mathcal S$ is also satisfiable.  We can then find a satisfying
  assignment of $\mathcal S$ by enumerating through all possible
  values $\boldsymbol x^\ast \in \mathbb Z^m$ and solving a linear program
  to decide whether there exists $\boldsymbol y^\ast \in \mathbb R^n$ such
  that $\boldsymbol x^*,\boldsymbol y^*$ satisfies $\mathcal S$.

  If the relaxed problem has no satisfying assignment then
  Proposition~\ref{prop:split} furnishes a finite set $\mathcal E$ of
  linear equations of the form $\boldsymbol u^\top \boldsymbol x = b$,
  with coefficients $\boldsymbol{u}\in\mathbb{Z}^m$ and
  $b\in \mathbb{Z}$, such that for any solution
  $\boldsymbol x^*\in \mathbb
  Z^m, \boldsymbol y^\ast \in \mathbb R^n$ of~\eqref{eq:bilinear} that
  has slack $\varepsilon$,
  the integer part $\boldsymbol x^\ast$
  satisfies an equation in $\mathcal E$.  We iterate through all such
  equations $\boldsymbol u^\top \boldsymbol x = b$ and in each case
  we apply Proposition~\ref{prop:semilinear} to write
  \[ \{ \boldsymbol x\in \mathbb Z^m : \boldsymbol u^\top \boldsymbol
    x =b,\, \boldsymbol x\geq \boldsymbol 0 \} \] as a linear set
  $L(B,P)$ for some finite set $B\subseteq \mathbb{Z}^m$ and matrix
  $P \in \mathbb{Z}^{m \times m-1}$.  Then for each vector
  $\boldsymbol w \in B$, we apply the change of variables
  $\boldsymbol x = \boldsymbol w + P\boldsymbol z$ to obtain a MIB
  system in one fewer integer variable to which we can recursively
  apply the procedure to determine satisfiability.

\end{proof}

In the following result we retrace the proof of
Theorem~\ref{thm:main-one}, this time keeping track of the size of the
integers involved.  We thereby obtain an upper bound on the smallest
satisfying assignment, showing that the gap satisfiability problem can
be solved in nondeterministic exponential time.

\begin{theorem}
  Consider a MIB system~\eqref{eq:standard} in which the integer
  constants have absolute value at most $H$.  If such a system is
  satisfiable with slack $\varepsilon$ then there is a satisfying
  assignment under which the integer variables have absolute value at
  most $2^{{\kappa_3}^{O(m^3(m+n))}}$, where
  $\kappa_3 := \left( \frac{mH^{m^2}}{\varepsilon} \right)$.
\label{thm:main-two}
\end{theorem}
\begin{proof}
  We first analyse the effect of a single variable-elimination step on
  the size of the integers in the system~\eqref{eq:standard}.  Recall
  that to eliminate an integer variable we assert a linear equation
  $\boldsymbol u^\top \boldsymbol x = b$, where
  $\| \boldsymbol u \| \leq \frac{2w(m)}{\delta}$ and
  $|b| \leq \kappa_2$.  Combining the lower bound
  $\delta \geq \frac{\varepsilon}{m^{1/2} H\kappa_1}$
  from~\eqref{def:delta}, the definition
  $\kappa_1:=m^{1/2} H^{(m^2+m)}$, the definition of $\kappa_2$
  in~\eqref{def:kappa2}, and the bound $\omega(m)=O(m^{3/2})$, we
  obtain that $ \| \boldsymbol u \|, |b| = \kappa_3^{O(1)}$, for
  $\kappa_3 := \left( \frac{mH^{m^2}}{\varepsilon} \right)$.
  
  Employing Proposition~\ref{prop:semilinear}, the equation
  $\boldsymbol u^\top \boldsymbol x = b,\, \boldsymbol x\geq
  \boldsymbol 0$, determines a substitution
  $\boldsymbol x = P\boldsymbol z + \boldsymbol w$ in which the
  elements of $P$ and $\boldsymbol w$ have absolute value at most
  $\kappa_3^{O(m)}$.  Since there are $m$ integer variables, the
  constants appearing over all MIB instances arising through the
  process of variable elimination have absolute value at most
  $\kappa_3^{O(m^2)}$.

  Consider a version of the relaxed system~\eqref{eq:relaxed2} in
  which the integer constants have magnitude at most $\kappa_3^{O(m^2)}$.
  For the purposes of our complexity analysis we augment the system
  with a new variable $r$ and constraints
  $r \geq \| \boldsymbol x \|+1$ and
  $r \geq \| \boldsymbol w \pm w(m)\boldsymbol u\|$ for each vertex
  $\boldsymbol w$ of the polyhedron $P(\boldsymbol y)$ (as defined
  in~\eqref{eq:defP}) and $\boldsymbol u \in U$.  The integer constants in
  the resulting system have absolute value at most
  $\kappa_3^{O(m^3)}$ by Hadamard's determinant inequality.  By construction, if
  $\boldsymbol x^*,\boldsymbol y^*$ is a satisfying assignment
  of~\eqref{eq:relaxed2} then the convex set
  $\{ \boldsymbol x \in P(\boldsymbol y^*) : \| \boldsymbol x\| \leq r\}$ has lattice width at
  most $w(m)$ and hence contains an integer point.  By
  Proposition~\ref{prop:semi-alg} an upper bound for $r$ is
  $2^{\kappa_3^{O(m^3(m+n))}}$, which concludes the proof.
\end{proof}

\begin{remark}
It is evident that the double exponential dependence of the magnitude
of the smallest satisfying assignment on the number of variables
in Theorem~\ref{thm:main-two} is unavoidable.  Indeed, consider the
following MIB system:
\begin{align*}
& x_iy_i \leq 1\quad (i=1,\ldots,n)\\
& x_{i+1}y_i \geq x_iy_0\quad (i=1,\ldots,n-1)\\
& x_1=2, y_0=1\\
& x_1,\ldots,x_n \in \mathbb{Z}_{\geq 0},\, y_0,\ldots,y_n \in
                 \mathbb{R}_{\geq 0}    
\end{align*}
Then any satisfying assignment satisfies
$x_{i+1} \geq \frac{x_i}{y_i} \geq x_i^2$ for $I=1,\ldots,n-1$, whence
$x_n \geq 2^{2^{n-1}}$.  The system moreover has a satisfying
assignment with slack $\varepsilon$ for any $\varepsilon>0$, obtained by
successively setting $y_i:=\frac{1+\varepsilon}{x_i}$ and
$x_{i+1}:= \lfloor \frac{x_i+\varepsilon}{y_i}\rfloor $ for $i=1,\ldots,n-1$.
\label{rem:lower}
\end{remark}

Proposition~\ref{prop:convex-hull} and Corollary~\ref{cor:bilinear}
give an exponential-time Turing reduction of the Gap Domination
Problem for MPTA to the Gap Satisfiability Problem for bounded MIB
systems, such that resulting instances of the Gap Satisfiability
Problem have size polynomial in that of the input MPTA.  We thus
obtain our second main result.
\begin{theorem}
The Gap Domination Problem for MPTA is decidable in non-deterministic
exponential time.
  \end{theorem}

\section{Conclusion}
Our main result shows that pareto curve of undominated reachable
observer values of a given MPTA can be approximated to arbitrary
precision.  This is in contrast with the situation for weighted timed
games, where it was recently shown that the optimal value of a
weighted timed game with positive and negative rates cannot be
computed to arbitrary precision~\cite{Guilmant24}.

Throughout this paper we have worked with MPTA with clock guards
defined by conjunctions of non-strict inequalities.  However, we claim
that for an MPTA $\mathcal A$ with guards comprising conjunctions of
both strict and non-strict inequalities, there exists an MPTA
$\mathcal A'$ with exclusively closed guards over the same set
$\mathcal Y$ of observers, such that every observer valuation
$\gamma \in \mathbb{R}^{\mathcal Y}$ reachable in $\mathcal A$ is also
reachable in $\mathcal A'$ and, conversely, for every valuation
$\gamma' \in \mathbb{R}^{\mathcal Y}$ reachable in $\mathcal A'$ and
every $\varepsilon>0$ there exists a valuation
$\gamma \in \mathbb{R}^{\mathcal Y}$ reachable in $\mathcal A$ such
that $|\gamma(c)-\gamma'(c)|<\varepsilon$ for all $c\in \mathcal Y$.
Indeed, such an MPTA $\mathcal A'$ is obtained by directly applying
the closure construction for timed automata in~\cite[Section
4]{OuaknineW03} to MPTA.  Then the ability to compute the pareto curve
of undominated reachable observer values of $\mathcal A'$ to arbitrary
precision allows one to achieve the same end for $\mathcal A$.

A direction for future work is to consider the feasibility of
approximate pareto analysis over infinite runs of MPTA.  For
double-priced timed automata, that is, MPTA with a single cost and
reward observer, it is known how to compute the optimal reward-to-cost
ratio over infinite computations using the
corner-point abstraction~\cite{Bouyer08c}.  For more general MPTA it is natural to consider
specifications that refer to multiple reward-to-cost ratios.

%In view of Theorem~\ref{thm:main-two} and Remark~\ref{rem:lower} it is
%natural to ask whether the gap satisfiability problem for bounded MIB
%systems is NEXPTIME-hard.  

%\paragraph*{Acknowledgements} 
%For the purpose
%of Open Access, the authors have applied a CC BY public copyright
%licence to any Author Accepted Manuscript (AAM) version arising from
%this submission.

%Andrew Scoones and James Worrell are supported by UKRI Fellowship EP/X033813/1.

%\section{Conclusion} \label{sec:conclusion}
%\input{Conclusion}

\bibliography{references}

\begin{thebibliography}{10}

\bibitem{AlurDill94}
R.\ Alur and D.\ Dill.
\newblock A theory of timed automata.
\newblock {\em TCS}, 126(2):183--235, 1994.

\bibitem{Alur01PTA}
R.~Alur, S.~La Torre, and G.~J. Pappas.
\newblock Optimal paths in weighted timed automata.
\newblock In {\em HSCC}, volume 2034 of {\em LNCS}, pages 49--62. Springer,
  2001.

\bibitem{AB06}
S.~Arora and B.~Barak.
\newblock {\em Computational Complexity: A Modern Approach}.
\newblock Cambridge University Press, 2006.

\bibitem{BhaveKT16}
Devendra B., Krishna S., and Trivedi A.
\newblock On nonlinear prices in timed automata.
\newblock In {\em Proceedings of the The First Workshop on Verification and
  Validation of Cyber-Physical Systems, V2CPS@IFM}, volume 232 of {\em
  {EPTCS}}, pages 65--78, 2016.

\bibitem{banaszczyk}
W.~Banaszczyk, A.~E. Litvak, A.~Pajor, and S.~J Szarek.
\newblock The flatness theorem for nonsymmetric convex bodies via the local
  theory of banach spaces.
\newblock {\em Mathematics of operations research}, 24(3):728--750, 1999.

\bibitem{Behrmann01}
G.~Behrmann, A.~Fehnker, T.~Hune, K.~G. Larsen, P.~Pettersson, J.~Romijn, and
  F.~W. Vaandrager.
\newblock Minimum-cost reachability for priced timed automata.
\newblock In {\em HSCC}, volume 2034 of {\em LNCS}, pages 147--161. Springer,
  2001.

\bibitem{Bouyer08b}
P.~Bouyer, E.~Brinksma, and K.~G. Larsen.
\newblock Optimal infinite scheduling for multi-priced timed automata.
\newblock {\em Formal Methods in System Design}, 32(1):3--23, 2008.

\bibitem{Bouyer08c}
P.~Bouyer, K.~G. Larsen, and N.~Markey.
\newblock Model checking one-clock priced timed automata.
\newblock {\em Logical Methods in Computer Science}, 4:1--28, 2008.

\bibitem{BouyerFLM11}
Patricia Bouyer, Uli Fahrenberg, Kim~G. Larsen, and Nicolas Markey.
\newblock Quantitative analysis of real-time systems using priced timed
  automata.
\newblock {\em Commun. {ACM}}, 54(9):78--87, 2011.

\bibitem{Brihayebruyereraskin06}
T.~Brihaye, V.~Bruy{\`{e}}re, and J.{-}F. Raskin.
\newblock On model-checking timed automata with stopwatch observers.
\newblock {\em Inf. Comput.}, 204(3):408--433, 2006.

\bibitem{DiakonikolasY09}
I.~Diakonikolas and M.~Yannakakis.
\newblock Small approximate pareto sets for biobjective shortest paths and
  other problems.
\newblock {\em {SIAM} J. Comput.}, 39(4):1340--1371, 2009.

\bibitem{FranzleSSW22}
M~Fr{\"{a}}nzle, M~Shirmohammadi, M~Swaminathan, and J~Worrell.
\newblock Costs and rewards in priced timed automata.
\newblock {\em Inf. Comput.}, 282:104656, 2022.

\bibitem{SwaminathanF09}
M.~Fr{\"a}nzle and M.~Swaminathan.
\newblock Revisiting decidability and optimum reachability for multi-priced
  timed automata.
\newblock In {\em The 7th International Conference on Formal Modelling and
  Analysis of Timed Systems}, pages 149--163. Springer Verlag, September 2009.

\bibitem{grotschel}
M.~Gr{\"o}tschel, L.~Lov{\'a}sz, and A.~Schrijver.
\newblock {\em Geometric algorithms and combinatorial optimization}, volume~2.
\newblock Springer Science \& Business Media, 2012.

\bibitem{Guilmant24}
Q.~Guilmant and J.~Ouaknine.
\newblock Inapproximability in weighted timed games.
\newblock In {\em Proceedings of CONCUR 24}, volume 311 of {\em LIPIcs}, 2024.

\bibitem{hemmecke2010nonlinear}
R.~Hemmecke, M.~K{\"o}ppe, J.~Lee, and R.~Weismantel.
\newblock {\em Nonlinear integer programming}.
\newblock Springer, 2010.

\bibitem{Henzinger98}
T.~A. Henzinger, P.~W. Kopke, A.~Puri, and P.~Varaiya.
\newblock What's decidable about hybrid automata?
\newblock {\em J. Comput. Syst. Sci.}, 57(1):94--124, 1998.

\bibitem{kannan1988covering}
R.~Kannan and L.~Lov{\'a}sz.
\newblock Covering minima and lattice-point-free convex bodies.
\newblock {\em Annals of Mathematics}, pages 577--602, 1988.

\bibitem{Khachiyan2000IntegerOO}
L.~Khachiyan and L.~Porkolab.
\newblock Integer optimization on convex semialgebraic sets.
\newblock {\em Discrete \& Computational Geometry}, 23:207--224, 2000.

\bibitem{Larsen01}
K.~G. Larsen, G.~Behrmann, E.~Brinksma, A.~Fehnker, T.~Hune, P.~Pettersson, and
  J.~Romijn.
\newblock As cheap as possible: Efficient cost-optimal reachability for priced
  timed automata.
\newblock In {\em CAV}, volume 2102 of {\em LNCS}, pages 493--505. Springer,
  2001.

\bibitem{Larsen08}
K.~G. Larsen and J.~I. Rasmussen.
\newblock Optimal reachability for multi-priced timed automata.
\newblock {\em TCS}, 390(2-3):197--213, 2008.

\bibitem{OuaknineW03}
J.~Ouaknine and J.~Worrell.
\newblock Revisiting digitization, robustness, and decidability for timed
  automata.
\newblock In {\em 18th {IEEE} Symposium on Logic in Computer Science {(LICS},
  Proceedings}, pages 198--207. {IEEE} Computer Society, 2003.

\bibitem{pottier91}
L.~Pottier.
\newblock Minimal solutions of linear diophantine systems: bounds and
  algorithms.
\newblock In {\em International Conference on Rewriting Techniques and
  Applications}, pages 162--173. Springer, 1991.

\bibitem{SchaeferS17}
M.~Schaefer and D.~Stefankovic.
\newblock Fixed points, nash equilibria, and the existential theory of the
  reals.
\newblock {\em Theory Comput. Syst.}, 60(2):172--193, 2017.

\bibitem{TollundJNTL24}
R.~G. Tollund, N.~S. Johansen, K.~{\O}. Nielsen, A.~Torralba, and K.~G. Larsen.
\newblock Optimal infinite temporal planning: Cyclic plans for priced timed
  automata.
\newblock In {\em Proceedings of the Thirty-Fourth International Conference on
  Automated Planning and Scheduling, {ICAPS}}, pages 588--596. {AAAI} Press,
  2024.

\bibitem{Zhang17}
Z.~Zhang, B.~Nielsen, K.~G. Larsen, G.~Nies, M.~Stenger, and H.~Hermanns.
\newblock Pareto optimal reachability analysis for simple priced timed
  automata.
\newblock In {\em {ICFEM}}, volume 10610 of {\em LNCS}, pages 481--495.
  Springer, 2017.

\end{thebibliography}

\end{document}